\documentclass{amsart}

\usepackage{epsfig}
\usepackage{latexsym}
\usepackage{amsmath}

\newtheorem{theorem}{Theorem}[section]
\newtheorem{proposition}[theorem]{Proposition}

\newtheorem{lemma}[theorem]{Lemma}

\newcommand{\F}{\mathcal F}
\newcommand{\C}{\mathcal C}
\newcommand{\T}{\mathcal T}

\newcommand{\PP}{{\mathbb P}}
\newcommand{\EE}{{\mathbb E}}

\newcommand{\old}[1]{{}}

\title[Is testing a tree easier?]{Phylogenetic information complexity: Is testing a tree easier than finding it?}
\author{Mike Steel, Laszlo Sz{\'e}kely and Elchanan Mossel}
\thanks{We thank the Isaac Newton Institute (Phylogeny programme)  for supporting this work.  LS acknowledges the support of NSF DMS 0701111, NIH NIGMS 1 RO1, GM078991-01, Marie Curie Fellowship HUBI MTKD-CT-2006-042794. EM acknowledges the support of an Alfred Sloan fellowship in Mathematics, NSF grants
DMS-0528488, DMS-0548249 (CAREER), and  DOD ONR grant
N0014-07-1-05-06}

\address{Biomathematics Research Centre, Department of Mathematics and
 Statistics, University of Canterbury, Christchurch, New Zealand. \\
 Department of Mathematics, University of South
  Carolina \\ Columbia SC, USA.
  \\Statistics Department, UC Berkeley, Berkeley, California, USA}

\email{m.steel@math.canterbury.ac.nz, szekely@math.sc.edu, mossel@stat.berekeley.edu}


\date{9 July 2008}


\begin{document}
\begin{abstract}
Phylogenetic trees describe the evolutionary history of a group of present-day species from a common ancestor.  These trees are typically reconstructed from aligned DNA sequence data. In this paper we analytically address the following question: is the amount of sequence data required to accurately reconstruct a tree significantly more than the amount required to test whether or not a candidate tree was the `true' tree?  By `significantly', we mean that the two quantities behave the same way as a function of the number of species being considered. We prove that,  for a certain type of model, the amount of information required is not significantly different; while for another type of model, the information required to test a tree is independent of the number of leaves, while that required to reconstruct it grows with this number. Our results combine probabilistic and combinatorial arguments.
\end{abstract}

\maketitle

Corresponding Author:

Mike Steel

Phone: +64-3-3667001, Ext. 7688

Fax: +64-3-3642587

Email: m.steel@math.canterbury.ac.nz

\newpage
\section{Introduction}

Phylogenetic trees are widely used in evolutionary biology to describe how species have evolved from a shared ancestral species.  In the last 25 years, aligned DNA sequence data and related sequences (amino acids, codons etc) have been widely used for reconstructing and analysing these trees \cite{fels, sem}.  Tree reconstruction methods usually assume that sequence sites evolve according to some Markov process. The question of how much data is required to reconstruct a phylogenetic tree has been considered by a number of biologists \cite{chu, lec, sai, wort}  and is topical, as it is not clear whether all trees for all taxa sets could be reconstructed accurately from the available data.

In earlier papers, we have analytically quantified the sequence length required for accurate tree reconstruction
when sites evolve i.i.d. under various Markov processes \cite{logs3, logs, cluster, mos1}. These bounds typically depend on the number of taxa and the properties of the tree -- in particular, how close the probability of a change of state (`substitution') on any edge is to $0$ or to its maximum possible value. It is the rate of growth in the sequence with the number of taxa that is of interest here. The growth rate in sequence length required for accurate tree reconstruction has a trivial lower bound growth of $\log(n)$  -- this comes from simply comparing the number of binary trees on $n$ leaves with the number of collections of $n$ sequences of given length.  What is perhaps surprising is that for certain finite-state Markov processes  this primitive rate of growth can be achieved for some models \cite{logs3}, given a bound on the substitution probabilities. This $\log(n)$ upper bound on sequence length also applies to a discrete infinite-state model (the `random cluster model') \cite{cluster}, given similar bounds on the substitution probabilities.  The $\log(n)$ behaviour for these two models changes to a polynomial dependence on $n$ when the probabilities of state change are allowed to exceed a certain critical value.

 In this paper, we address a quite different question: namely if one has both the data and a proposal for a `true' tree (i.e. the tree that produced the data under the model) we would like to test whether this tree is indeed the true tree, or whether some other tree must have produced the data. The answer provided must be correct with high probability.   This concept of
testing fits into the Popperian tradition -- we would like to be able to  refute the hypothesis that a particular tree produced the data, without necessarily exhibiting the tree that did.  In statistics, the theory of testing among a discrete set of hypotheses has a long history (see, for example, Wald's  paper from 1948 \cite{ber}). In contemporary phylogenetics, the concept of testing a tree is timely, as various `tree of life' projects begin to provide detailed, large candidate evolutionary trees, the question of using data to `test' any such  candidate tree arises.

In this paper, we ask whether the information (sequence length) required for these two tasks -- reconstructing versus testing -- is fundamentally the same, i.e. that it grows at the same rate as a function of the size of the taxon set.  Intuitively, testing should be `easier' (require less data) than reconstructing, since for testing, one has additional information, and one is simply asked to make a binary decision.  This suggested difference is somewhat analogous to the ``$P \neq NP$'' conjecture in computational complexity,  where any candidate solution for a problem in the class $NP$  can be readily verified or refuted (in polynomial time) but it is believed that finding such a solution is fundamentally harder (i.e. not always possible in polynomial time). Of course, in our setting, we are dealing with sequence length, not computing time, but the two problems have an  analogous flavour.

In the next section, we describe a general framework for discussing these issues, and we exhibit an (abstract) example where the sequence length required to test a discrete parameter is far less than the sequence length required to reconstruct it.  Turning to the phylogenetic setting in Section~\ref{phy}, we show that when sequence sites evolve i.i.d. according to a finite-state Markov process then
 testing requires sequence length growth of rate $\log(n)$ -- which is the same as reconstructing requires, given bounds on the substitution probabilities.   By contrast,  for a discrete infinite-state Markov process, the situation is quite different -- constant length sequences suffice to test, but order $\log(n)$ length sequences are required for reconstruction.  We conclude the paper with some brief comments.

\section{Testing versus reconstructing}

In this section, we describe definitions and properties of  testing and reconstructing in a general setting; we will specialize our approach to the phylogenetic setting in the following section.

Suppose $A=A_n$ and $U= U_n$ are finite sets, and that we have a random variable $X=X_{(a,\theta)}$ taking values in $U$ and whose distribution  depends on the discrete parameter $a\in A$, and perhaps some nuisance parameter $\theta$ taking values in a set $\Theta(a)$ (in the next section, $A$ will be a set of trees, $U$ will be a set of site patterns and the nuisance parameters will be edge lengths of the tree -- all these concepts will be described later).  We call $X_{(a, \theta)}$ a {\em parameterized random variable}, and when  the nuisance parameter is either absent or has been specified for each $a$ (so that the distribution of $X_{(a,\theta)}$ just depends on $a$), we refer to a {\em simply-parameterized random variable}.

Given a sequence  ${\bf u}=  (u_1, u_2, \ldots, u_k)$  of $k$ i.i.d. observations of $X_{(a, \theta)}$,
we would like to use ${\bf u}$  to identify the discrete parameter $a$ correctly with high probability.   This reconstruction task is always possible for sufficiently large values of $k$ provided a weak `identifiability' condition holds, namely that for all $a \in A$, and
$\theta \in \Theta(a)$ we have:
$$\inf_{a' \neq a, \theta' \in \Theta(a')} d((a,\theta), (a', \theta')) >0$$
(see \cite{inv3}),
where $d((a,\theta), (a', \theta'))$ denotes the $l_1$ distance between the probability distribution of the random variables $X_{(a, \theta)}$ and $X_{(a', \theta')}$.

The two tasks that we consider can be summarized, informally, as follows.

\begin{itemize}
\item[]  {\bf Reconstructing:} Given ${\bf u} \in U^k$, determine with high probability the value $a \in A$ that generated ${\bf u}$ (for some $\theta \in \Theta(a)$).
\item[] {\bf Testing:} Given a candidate value $a\in A$, as well as ${\bf u} \in U^k$, determine with high probability whether or not ${\bf u}$ was generated by $(a, \theta)$ (for some $\theta \in \Theta(a))$.
\end{itemize}

We are interested in determining and comparing the number of i.i.d. samples required to carry out these tasks. Clearly
testing is `easier'  than reconstructing (i.e. testing requires a smaller or equal value of $k$ than reconstructing for the same accuracy), since one can always test by reconstructing and comparing the reconstructed object with the candidate object. Thus we will be particularly interested in determining whether the value of $k$ grows at the same rate with $n$ for these two tasks.

In general a basic lower limit on $k$ is required for reconstructing, and a much weaker one for testing, which shows that testing can require asymptotically much shorter sequences.  Before describing these bounds (Proposition~\ref{two}), we formalize the concept of reconstructing and testing.

\subsection{Definitions: Reconstruction, testing and accuracy}

Throughout we will let $(X_{(a, \theta)}^1, \ldots X_{(a, \theta)}^k)$ denote a sequence of $k$ i.i.d. observations generated by $(a, \theta)$.

A {\em reconstruction method} $R$ is a random variable\footnote{Reconstruction is often viewed as a deterministic function, but in reality, most methods have to break ties and so allowing $R$ to be random allows this (and perhaps other generalities) -- also, we will always assume that reconstruction and generation are stochastically independent processes.} $R({\bf u})$ taking values in $A$ and which depends on
${\bf u} \in U^k$.
Then: $$\rho_{(a,\theta)} ^R:= \PP(R((X_{(a, \theta)}^1, \ldots X_{(a, \theta)}^k))= a)$$
is the probability that $R$ will correctly reconstruct $a$ from $k$ i.i.d. samples generated by $(a,\theta)$.
We say that $R$ has {\em (reconstruction) accuracy $1-\epsilon$ (for $k$ samples)} if, for all $a \in A$
 and $\theta \in \Theta(a)$, we have,
$$\rho_{(a, \theta)}^R \geq 1-\epsilon.$$

Note that, if we let  $p_{a, \theta}(u) = \PP(X_{(a,u)} = u)$, and for ${\bf u} \in U^k$ let\footnote{We will write $p_{a}({\bf u})$ in the case of a simply-parameterized random variable.}
$$p_{a, \theta}({\bf u}) := \prod_{i=1}^kp_{a, \theta}(u_i),
$$
then: $$\rho_{(a,\theta)}^R=\sum_{{\bf u} \in U^k} p_{(a,\theta)}({\bf u } ) \cdot \PP(R({\bf u}) =a).$$

A {\em testing process} $\psi$ is a collection of random variables $ \psi(a, {\bf u}): a \in A, {\bf u} \in U^k$ taking values in $ \{{\rm true}, {\rm false}\}$.  In the case where $\psi$ assigns `true' or `false' with probability $1$ (for each choice of $(a, {\bf u})$), we
say that $\psi$ is a {\em deterministic testing process}.

Let $\epsilon>0$.  We say that a testing process has {\em accuracy $1-\epsilon$ (for $k$ samples)} if for all $a \in A$, and $\theta\in \Theta(a)$, the following
two conditions hold:
\begin{equation}
\label{eps1}
\PP(\psi(a, (X^1_{(a,\theta)}, \ldots, X^k_{(a,\theta)})) = {\rm true}) \geq 1- \epsilon,
\end{equation}
and for any $b \neq a$, and $\theta' \in \Theta(b)$,
$$\PP(\psi(a, (X^1_{(b,\theta')}, \ldots, X^k_{(b,\theta')})) = {\rm false}) \geq 1-\epsilon.$$

In other words, $\psi$ returns `true' with probability at least $1-\epsilon$ when the discrete parameter $a \in A$ is
tested against the data it produced, and $\psi$ returns `false' with probability at least $1-\epsilon$
when any other particular element of $A$ is tested against the data. Note that any collection of random variables $X_{(a, \theta)}$ has a trivial testing process with accuracy $\frac{1}{2}$; namely for each $a \in A$ and ${\bf u} \in U$ let
$\psi(a, {\bf u}) = {\rm true}$ with probability $\frac{1}{2}$. To exclude such trivialities we will generally assume that $\epsilon<\frac{1}{2}$.

If Inequality (\ref{eps1}) is strengthened to:
$$\PP(\psi(a, (X^1_{(a,\theta)}, \ldots, X^k_{(a,\theta)})) = {\rm true}) =1$$
for all $a \in A$ and $\theta \in \Theta(a)$, we say that the testing process has {\em strong accuracy} $1-\epsilon$.

The following proposition shows that reconstructing in general can require considerably longer sequences than testing.
\bigskip

\begin{proposition}
\label{two}
\begin{itemize}
\mbox{}
\item[(i)]
Suppose a simply-parameterized random variable $X_a$ ($a \in A$) takes values in $U$ and has a reconstruction method with accuracy strictly greater than $\frac{1}{2}$ for $k$ samples.
Then $|A| \leq |U|^k$, and so $k \geq \frac{ \log(|A|)}{\log(|U|)}$, or, equivalently:
$$\log(|U|) \geq \frac{1}{k}\log(|A|).$$
\item[(ii)]
For any $\epsilon>0$ and $k=1$,  there exist sets $A, U$ and a simply-parametized random variable $X_{a}$ (for $a \in A$)  taking values in $U$, and a deterministic testing procedure $\psi$ that has strong accuracy of $1-\epsilon$, such that: $$\log(|U|) = O(\log(\log(|A|))).$$
\end{itemize}
\end{proposition}
\begin{proof}
Part (i) was established in \cite{inv1} (Theorem 2.1 (ii)).
For Part (ii), let $U = \{1, \ldots, n\}$ and let $A$ be a collection of subsets of $U$ with the property that
for any two elements $a,a' \in A$:
\begin{equation}
\label{fracs}
\frac{|a \cap a' |}{\min \{|a|, |a'|\}} \leq \epsilon.
\end{equation}
Consider the following simply-parameterized random variable $X_a (a \in A)$ defined by the rule that $X_a$ selects one of the elements of $a$ uniformly at random. Note that $X_a$ takes values in the set $U$.  Consider the deterministic testing process $\psi$ defined by the rule:
$$\psi(a,u) =
\begin{cases}
{\rm true}, & \mbox{if  $u \in a$;}\\
{\rm false}, & \mbox{otherwise}.
\end{cases}$$
Clearly, $$\PP(\psi(a, X_a) = {\rm true}) = 1,$$
and Condition (\ref{fracs}) ensures that for $b \neq a$:
 $$\PP(\psi(a, X_b) = {\rm true}) \leq \epsilon,$$
and so $\psi$ has strong accuracy $1-\epsilon$.  Thus it remains to construct a family $A$ of subsets of $\{1, \ldots, n\}$ satisfying (\ref{fracs}) and of cardinality at least $e^{n^\eta}$ for some $\eta>0$ (since in that case $\log(|U|) = O(\log(\log(|A|)))$).

The existence of such a large collection can be established by using the probabilitistic method as follows.  Construct $N$ random subsets of $\{1, \ldots, n\}$ by the following process: for each element $i$ of $\{1, \ldots, n\}$ place $i$ in the set with probability $n^{-2/3}$; otherwise, leave that element out. Using standard results on the asymptotic distribution of sums of i.i.d. random variables,  the probability $p$ of the event that (i) all of the $N$ sets are of size at least $n^{0.3}$, and (ii) that all pairs of sets intersect in at most $n^{0.2}$ points satisfies (by the subadditivity of probability):
 $$p \geq 1-N\exp(-n^{c_1}) - \binom{N}{2}\exp(-n^{c_2}),$$for positive constants $c_1, c_2$. Now, for $N = e^{n^\eta}$ where $0<\eta<\min\{c_1, c_2\}$  it holds that $p>0$ for sufficiently large values of $n$. In  this case, a collection of sets must exist that satisfy Conditions (i) and (ii).  Finally, this collection will also satisfy (\ref{fracs}) provided $n$ is large enough that $\frac{n^{0.2}}{n^{0.3}} < \epsilon$. This completes the proof. \end{proof}

Given a reconstruction method $R$,  a canonical testing process $\psi_{R}$ is associated
with it, defined as follows:
$$\psi_{R}(a, {\bf u} ) = {\rm true}  \Leftrightarrow R({\bf u})=a.$$
for all $a \in A$ and ${\bf u} \in U^k$.
The following lemma follows easily from the definitions.
\begin{lemma}
\label{oneway}
Given a parameterized random variable $X_{(a, \theta)}$ and an associated reconstruction method $R$ with accuracy $1-\epsilon$, the associated testing process $\psi_{R}$ has reconstruction accuracy $1-\epsilon$.  \end{lemma}

In summary, it is clear that in general, `testing' can require considerably less information (sequence length) than `reconstructing.'  We now consider what happens in a specific setting that arises in computational evolutionary biology.

\section{Testing versus reconstruction in phylogenetics}
\label{phy}

A {\em phylogenetic ($X$--) tree} is a tree $\T$, whose leaf set $X$ is labelled and whose interior vertices are unlabeled and of degree at least $3$. If, in addition, every interior vertex of $\T$ has degree exactly $3$ then $\T$ is said to be {\em binary}.  Without loss of generality, we can usually take $X=\{1, \ldots, n\}$.

In this section, we specialize, letting $A= A_n$  be the set of binary phylogenetic trees on leaf set $\{1,\ldots, n\},$ and letting $U=U_n$ be the set of
site patterns on the leaf set $\{1, \ldots, n\}$ generated under some Markov process on the tree.  Let $k_t(n)$ denote the sequence length required to test a phylogenetic tree with high probability and let $k_r(n)$ be the sequence length required to reconstruct a tree with high probability  (under the same model).

For one class of models (finite-state Markov processes with an irreducible rate matrix), we will show (Theorem~\ref{finitestate}) that $k_t(n)$ grows at least logarithmically with $n$ (even in the simply-parametric setting, where each tree has a fixed set of edge lengths).   It had already been established earlier that $k_r(n)$ grows at least logarithmically for very general models of sequence evolution (and in some more restricted models and parameter sets, grows at least polynomially) \cite{logs3, logs}.  Thus, when $k_r(n)= O(\log n)$, we have $k_t(n)=\Theta(k_r(n))$.

However for a closely-related Markov process - the `random cluster model', which can be used to model rare genomic events, the situation is surprisingly different in one respect. Although reconstruction still requires at least logarithmic (and in a certain range polynomial) number of samples \cite{cluster}, testing with strong accuracy of $1-\epsilon$ can be achieved with $O(1)$ samples (Theorem~\ref{randyclust}). Thus, in this case we  have  $k_t(n) = o(k_r(n))$. Moreover this applies even in the simply-parameterized setting (where each tree has a fixed set of edge lengths).

We will now describe these results, beginning with the finite-state model.

\subsection{Testing for a finite-state Markov process requires at least $\log(n)$ length sequences}

Finite-state Markov processes on trees underlie many approaches in molecular phylogenetics (see, for example, \cite{fels}).   We provide a brief formal description; for more details, the reader may wish to consult \cite{fels} or \cite{sem}.

\begin{figure}[ht]
 \label{figure1}\begin{center}
\resizebox{12.5cm}{!}{
\input{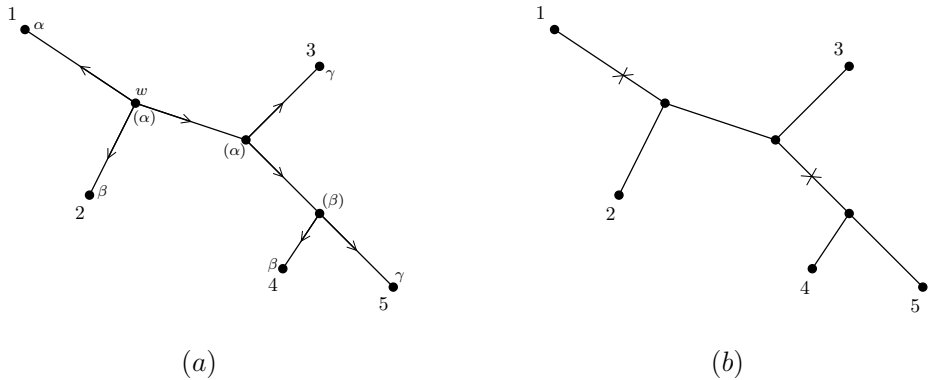}
}
\end{center}
\caption{(a): In a finite-state Markov process, a random state ($\alpha$) at some vertex ($w$) evolves along the arcs of the tree (directed away from $w$). This gives rise to states at the internal vertices of the tree (in this example $(\alpha), (\alpha), (\beta)$) and a site pattern at the leaves (in this case $\alpha,\beta,\gamma, \beta, \gamma$ where the ordering corresponds to leaf order $1, 2, 3, 4, 5$).
(b):  In the random cluster model each edge $e$ is independently cut (indicated by $\times$) with an associated probability $p(e)$. In this example two edges were cut resulting in the leaf partition (character) of $\{\{1\}, \{2,3\}, \{4,5\}\}$.}
\end{figure}

A finite-state Markov process is a continuous-time Markov process whose state space is some finite set; we will denote the rate (intensity) matrix of this process by $R$.  We assume that $R$ forms a reversible Markov process and we let $\pi$ denote the equilibrium distribution on the states (determined by $R$).

Now, suppose we have a phylogenetic $X$--tree $\T$ for which each edge $e$ has some strictly positive valued `length' ($l(e)$).   In this way, we can define a Markov process on $\T$  as follows ({\em c.f.} Fig. 1(a)).  To some vertex $w$, we assign states according to the distribution $\pi$. Then assign states to the remaining vertices of the tree by orienting the edges of $\T$ away from $w$; for each arc $(u,v)$ for which $u$ has been assigned a state $s$, assign to $v$ the state obtained by applying the continuous-time Markov process with initial state $s$ for duration $l(e)$.  Thus the transition matrix associated to edge $e$ is $\exp(Rl(e))$ and the joint probability distribution on the vertices of $\T$ is independent of the choice of the initial vertex $w$ (by the reversibility assumption).   Such a model induces  a marginal distribution on the set of site patterns  -- assignments of states to the elements of $X$ (the leaves of $\T$) -- and this constitutes a single sample of the process (the site pattern is a random variable parameterized by the pair ($\T, l)$, where $l$ assigns the lengths to the edges of $\T$).

The main result of this section is the following. Recall that a rate matrix is {\em irreducible} if the probability of a transition from any one state to any other state in time $\delta>0$ is strictly positive.

\begin{theorem}
\label{finitestate}
Suppose we have a finite-state Markov process, with an irreducible rate matrix, on binary phylogenetic trees with leaf set $\{1, \ldots, n\}$. Suppose that we generate $k$ i.i.d. site patterns.  Then any  testing procedure that has accuracy $> \frac{1}{2}$ requires $k$ to grow at least at the rate $\log(n)$,
even in the simply-parameterized setting where all edge lengths are equal to a fixed strictly positive value.
\end{theorem}

{\bf Remark:}
Theorem~\ref{finitestate} should be viewed alongside the result of \cite{logs3}, which shows that tree reconstruction (under the 2-state symmetric Markov model) with high accuracy is possible for sequences of length order $\log(n)$, even when the edge lengths are not known but constrained to lie within
a fixed interval $[f,g]$ for any $f>0$ and when $g$ is sufficiently small.

To establish Theorem~\ref{finitestate}, we first require a general result.
\begin{lemma}
\label{testconnect}
Suppose a simply-parameterized random variable $X_{a}$ has a testing process $\psi$ with accuracy $1-\epsilon$ for $k$ samples. Then:
\begin{itemize}
\item[(i)]
$d^{(k)}(a,a') \geq 2(1-2\epsilon)$ for all $a, a' \in A$ with $a \neq a'$,
$$\mbox{ {\rm where} }  d^{(k)}(a,a') := \sum_{{\bf u} \in U^k}|p_a({\bf u}) - p_{a'}({\bf  u})|$$
\item[(ii)] Let $A'$ be a proper, nonempty subset of $A$, and let $a \in A-A'$. Consider the following random variable $X'$ that is simply parameterized by the set $\{a, *\}$,
and  defined as follows:  $X'_a= X_a$, and $X'_* = X_Y$ where $Y$ is selected
uniformly at random from the nonempty set $A'$. Then $X'$ has a reconstruction method with  accuracy $1-\epsilon$ for $k$ samples. In particular:
$$d^{(k)}(a, *) \geq 2(1-2\epsilon).$$
\end{itemize}
\end{lemma}

\begin{proof}
{\em Part (i).} First observe that $d^{(k)}(a,a')$ is twice the variational distance between the probability distributions $p_a$ and $p_{a'}$ on $U^k$, i.e.: $$d^{(k)}(a,a') = 2\cdot \max_E|\PP_a(E)-\PP_{a'}(E)|,$$ where maximization is over all events $E$ on $U^k$.
For each $a,a' \in A$, let $E_{a,a'}$ be the event that $\psi(a, (X_{(a, \theta)}^1, \ldots X_{(a, \theta)}^k)) = {\rm true}.$ Then $E_{a,a'}$ has probability at least $1-\epsilon$ when $a= a'$ and probability at most $\epsilon$ when $a\neq a'$.  Consequently,
$d^{(k)}(a,a') \geq 2|\PP_a(E_{a,a'}) - \PP_{a'}(E_{a,a'})| \geq 2(1-2\epsilon)$.

{\em Part (ii).}  Let $R: U^k \rightarrow \{a, *\}$ be defined as follows:
$$R({\bf u}) =
\begin{cases}
a, & \mbox{if $\psi(a, {\bf u}) =$ true;}\\
*, & \mbox{if $\psi(a, {\bf u}) =$ false.}
\end{cases}$$
Then,
$\rho_a^R = \PP(\psi(a, (X_{(a, \theta)}^1, \ldots X_{(a, \theta)}^k)) = {\rm true})  \geq 1-\epsilon$.
Moreover:
$$\rho^R_*= 1- \PP(\psi(a, (X_{(Y, \theta)}^1, \ldots X_{(Y, \theta)}^k)) = {\rm true}),$$ and:
$$\PP(\psi(a, (X_{(Y, \theta)}^1, \ldots X_{(Y, \theta)}^k)) = {\rm true}) = \sum_{y \in A'}
 \PP(\psi(a, (X_{(y, \theta)}^1, \ldots X_{(y, \theta)}^k)) = {\rm true}) \cdot \frac{1}{|A'|}.$$
By assumption, each term in the sum is $\leq \epsilon$ and so $\PP(\psi(a, (X_{(Y, \theta)}^1, \ldots X_{(Y, \theta)}^k)) = {\rm true}) \leq \epsilon$.
Thus, $\rho_*^R \geq 1-\epsilon$, as required.
By Lemma~\ref{oneway} there is a testing procedure for $X'$ that has accuracy at least $1-\epsilon$ and
so the second claim in Part (ii) now follows by Part (i).
\end{proof}

\begin{figure}[ht]
 \label{figure2}\begin{center}
\resizebox{8cm}{!}{
\input{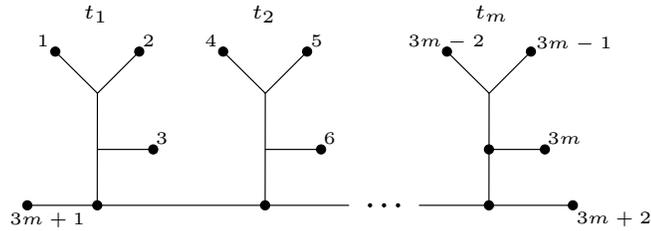}
}
\end{center}
\caption{The generating tree in the proof of Theorem~\ref{finitestate}.}
\end{figure}

\noindent{\em Proof of Theorem~\ref{finitestate}}
Let $a$ be the tree shown in Fig.~2 (for convenience, we will take all the edge lengths in this tree to be equal). With a view to applying Lemma~\ref{testconnect}, let $A'$ denote the set of $m$ trees consisting of precisely those trees obtained from $a$ by interchanging the leaf labels $3i-1$ and $3i$ for one value of $i=1, \ldots, m$ (and keeping the edge lengths fixed). Suppose there exists
a testing process of accuracy $1-\epsilon$ for phylogenetic trees. Then, by Lemma~\ref{testconnect}(ii)
we have:
\begin{equation}
\label{stepzero}
d^{(k)}(a, *) \geq 2(1-2\epsilon).
\end{equation}
Suppose we generate $k$ sites i.i.d. under parameter $a$ or $*$ (in the case of $*$, we select the random element of $A'$ and then generate $k$ sites i.i.d. using that element).
Let $C$ be the random vector variable that lists the sequences occurring at the  unlabelled internal vertices of the model tree $a$ in Fig. 2.
Note that $C$ has the same probability distribution for any element $b$ in $A'$ as it does for $a$, and so, in particular, we have:
$$\PP_a(C=c) = \PP_*(C=c)$$
for all choices of $c$.
Thus: $$d^{(k)}(a, *) = \sum_{{\bf u} \in U^k}|\PP_a({\bf u})-\PP_*({\bf u})| = \sum_{\bf u}|\sum_c(\PP_a({\bf u}|c)-\PP_*({\bf u}|c))\cdot \PP_a(C=c)|,$$
and so:
\begin{equation}
\label{firstep}
d^{(k)}(a, *) \leq \left(\max_c \sum_{\bf u}|\PP_a({\bf u}|c)-\PP_*({\bf u}|c))|\right)\cdot \sum_c \PP_a(C=c) = \max_c \sum_{\bf u}|\PP_a({\bf u}|c)-\PP_*({\bf u}|c))|.
\end{equation}
We will establish the following crucial inequality:
For any $c$:
\begin{equation}
\label{secondstep}
\sum_{{\bf u} \in U^k}|\PP_a({\bf u}|c)-\PP_*({\bf u}|c))| < \tau^k/\sqrt{m}
\end{equation}
 for a constant $\tau$.
Then, combining (\ref{secondstep}), (\ref{firstep}) and
(\ref{stepzero}) gives $\tau^k/\sqrt{m} > 2(1-2\epsilon)$ and so $k$ must grow at least logarithmically with $n$ ($=3m+2$), as claimed by the Theorem.

Thus, to establish the theorem, it suffices to justify Inequality (\ref{secondstep}).

For any $c$ maximizing $\sum_{\bf u}|\PP_a({\bf u}|c)-\PP_*({\bf u}|c))|$, let us denote
$\PP_x({\bf u}|c)$ by $Q_x({\bf u})$ for $x=a, *$ and $b \in A'$.   Then:
\begin{equation}
\label{hob1}
\sum_{\bf u}|Q_a({\bf u})-Q_*({\bf u})| = \sum_{\bf u}|Q_a({\bf u})-\frac{1}{m}\sum_{b \in A'} Q_b({\bf u})|.
\end{equation}
We can rewrite the expression on the right-hand side of (\ref{hob1}) as:
\begin{equation}
\label{fourthstep}
|\sum_{\bf u} Q_a({\bf u})(1- \frac{1}{m}\sum_{b \in A'}\frac{Q_b({\bf u})}{Q_a({\bf u})})| \leq \sum_{\bf u} Q_a({\bf u}) \cdot |\frac{1}{m}\sum_{b \in A'}(1-\frac{Q_b({\bf u})}{Q_a({\bf u})})|.
\end{equation}
In particular,  consider the following random variable: $$Z_b:= 1-\frac{Q_b(\xi)}{Q_a(\xi)}$$ for each $b \in A'$, where $\xi$ is the random element of $U^k$ generated by the probability distribution $Q_a$. Then
the expression on the right-hand side of (\ref{fourthstep})
can be rewritten as: $$\EE\left(|\frac{1}{m}\sum_{b \in A'} Z_b|\right),$$ where expectation is taken with respect to the
probability distribution $Q_a$ on $U^k$.
Now, for all $b \in A'$ we have:
$$\EE(Z_b) = \sum_{{\bf u} \in U^k} Q_a({\bf u})\cdot \left(1-\frac{Q_b({\bf u})}{Q_a({\bf u})}\right) = \sum_{{\bf u} \in U^k} (Q_a({\bf u})-Q_b({\bf u}))=1-1=0,$$
 and the $Z_b$ are independent random variables (since we have conditioned on a particular value $C=c$).  Thus, by Jensen's inequality:
\begin{equation}
\label{hob2}
\EE\left(|\frac{1}{m}\sum_{b \in A'} Z_b|\right)^2 \leq \EE\left((\frac{1}{m}\sum_{b \in A'} Z_b)^2\right) = \frac{1}{m^2}\sum_{b \in A'} \EE(Z_b^2).
\end{equation}

Now,  since $\EE(Z_b)=0$, we have $1 + \EE(Z_b^2) = \EE((Z_b+1)^2) = \EE(\frac{Q_b(\xi)^2}{Q_a(\xi)^2})$ and so:
\begin{equation}
\label{fifthstep}
\EE(Z_b^2) = -1 + \EE\left(\frac{Q_b(\xi)^2}{Q_a(\xi)^2}\right).
\end{equation}
Also, writing ${\bf u} = (u_1, \ldots, u_k)$ and $c= (c_1, \ldots, c_k)$, the irreducibility condition of the Markov process ensures that the following inequality holds:
 \begin{equation}
\label{tau}
\frac{\PP_b(u_i|c_i)}{\PP_a(u_i|c_i)} \leq \tau,
\end{equation}
 for some absolute constant $\tau$ dependent only on the rate matrix $R$ and the (equal) value of the common edge length.
Thus, by independence, $\frac{Q_b({\bf u})}{Q_a({\bf u})} \leq \tau^k$, and so (\ref{fifthstep}) gives:
$$\EE(Z_b^2) \leq -1 + \tau^{2k} < \tau^{2k}.$$
Consequently, by (\ref{hob2}):
$$\EE\left(|\frac{1}{m}\sum_{b \in A'} Z_b|\right) \leq \sqrt{ \frac{1}{m^2} |A'|\tau^{2k} } < \tau^k/\sqrt{m},$$
and so (by (\ref{hob1}) and (\ref{fourthstep})) we have verified Inequality (\ref{secondstep}), and thereby completed the proof.
\hfill $\Box$

\subsection{An O(1) test for the random cluster model}

In this section, a {\em character} (on $X$) denotes an arbitrary partition $\{\alpha_1, \ldots, \alpha_m\}$ of $X$ into any number of disjoint subsets.

In the {\em random cluster model}, one has a phylogenetic $X$--tree $\T$ and each edge $e$ has  an associated probability $p(e)$ that the edge of $\T$ is cut. These cuts are performed independently across the tree, resulting in a (generally disconnected) graph and the leaves in each connected component form the blocks of a resulting random character on $X$. Thus $\T$ along with the $p(e)$ values (called `substitution probabilities')  provide a well-defined probability distribution on characters on $X$ (see Fig. 1(b)).

This is the same distribution on characters as one obtains in the limit as $s \rightarrow \infty$ from a finite-state Markov process on $\T$ that has an $s \times s$ rate matrix in which all its off-diagonal entries are equal, and where one considers the character on $X$ whose blocks are the sets of leaves of the same state.  Thus we can view the random cluster model as a type  of infinite-state Markov process.  The model is relevant for describing evolution in settings where transitions generally lead to states that have not appeared elsewhere in the tree (such as with gene order re-arrangement, or other rare genomic events).

Given a character $\chi$ on $X$ and a phylogenetic $X$--tree $\T$, let $\T_{\alpha}$ denote the minimal subtree of $\T$ connecting the leaves of $\T$ in block $\alpha$. Then $\chi$ is said to be {\em homoplasy-free}  on $\T$  if the collection of trees $\T_\alpha: \alpha \in \chi$ is vertex-disjoint.
Given a sequence $\C= (\chi_1, \chi_2, \ldots, \chi_k)$ of characters, consider the following deterministic testing process $\psi_H$ on phylogenetic trees:

$$\psi_H(\C, \T) =
\begin{cases}
{\rm true}, & \mbox{if $\chi_i$ is homoplasy-free on $\T$ for $i=1, \ldots, k$;}\\
{\rm false}, & \mbox{otherwise.}
\end{cases}$$

 \begin{figure}[ht] \begin{center} \label{figure3}
\resizebox{8cm}{!}{
\input{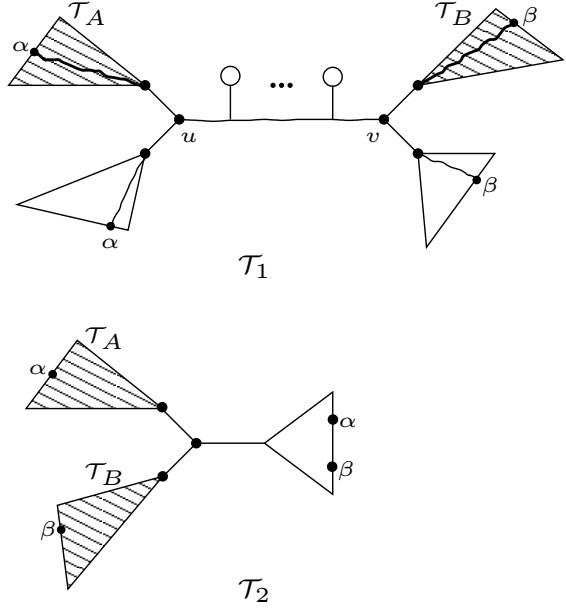}
}
\end{center}
\caption{Canonical representation of two different binary phylogenetic trees in the proof of Theorem~\ref{randyclust}}
\end{figure}

\begin{theorem}
\label{randyclust}
Under the random cluster model  on binary phylogenetic trees with leaf set $\{1, \ldots, n\}$, suppose that we generate $k$ i.i.d. characters, where the substitution probability $p(e)$  on any edge $e$ lies in some fixed interval $[a,b]$ where $0<a \leq b<\frac{1}{2}$.  Then $\psi_H$ is a testing process with
strong accuracy of $1-\epsilon$ whenever the number of characters is at least $k$, where:
$$k = \gamma^{-1}\cdot\log(\epsilon^{-1}),$$
and where $\gamma = a\cdot \frac{(1-2b)^4}{(1-b)^4}.$ This holds for all values of $n$.
\end{theorem}

To prove Theorem~{\ref{randyclust} we first require a combinatorial lemma.
\begin{lemma}
\label{decomplem}
Suppose $\T_1, \T_2$ are two binary phylogenetic $X$--trees, with $\T_1 \neq \T_2$, and $|X|>3$.
There exist induced rooted phylogenetic subtrees $\T_A$, $\T_B$, on leaf sets $A$ and $B$, respectively, where $A, B$ are disjoint, nonempty subsets of $X$, such that:
\begin{itemize}
\item[(i)] $\T_A$ and $\T_B$ are present as pendant subtrees in $\T_1$ and $\T_2$; and
\item[(ii)]  the root of $\T_A$  and of $\T_B$ are adjacent to a common vertex in $\T_2$ but not in $\T_1$.
\end{itemize}
\end{lemma}
\begin{proof}
The proof is by induction on $|X|$.  For $|X|=4$, the result is easily seen to hold. Suppose the lemma holds for $|X|=n-1$ where $n\geq 5$, and
that $|X|=n$. If $\T_2$ has a cherry (a pair of elements $\{x,y\}$ of $X$ that label leaves that are adjacent to a common vertex) that is not also a cherry of $\T_1$, then we can take $A=\{x\}, B=\{y\}$ and the claim in the theorem holds.  Otherwise, every cherry of $\T_2$ is a cherry of $\T_1$ and since $\T_2$ has at least one cherry ( say $\{x,y\}$ ) we may replace
$\T_1, \T_2$ with the pair of trees $\T_1', \T_2'$ obtained by deleting  the leaves labeled $\{x,y\}$ and their incident edges from each tree, and assigning each newly created leaf vertex the label $v_{x,y}$. Note that
$\T_1', \T_2'$ are binary phylogenetic trees, with a leaf label set $X-\{x,y\} \cup\{v_{x,y}\}$, and that
$\T_1' \neq \T_2'$ (otherwise it is easily seen that $\T_1=\T_2$). Thus we may apply the induction hypothesis to the pair
$\T_1', \T_2'$.  Given the sets $A$, $B$ for this pair that meet the requirements stated in the lemma, we then
replace any occurrence of $v_{x,y}$ in $A$ or $B$
by the elements $x,y$ -- the resulting pair of sets now satisfies the requirements stated in the lemma. This completes the proof.
\end{proof}

\noindent{\em Proof of Theorem~\ref{randyclust}}
Suppose that $\T_1$ is the binary phylogenetic $X$--tree that generates the characters. Then $\psi_H(\F, \T_1) = `{\rm true}$' with probability $1$, since the event $\chi$ is homoplasy-free on  $\T_1$  has probability $1$ for any character $\chi$ that evolves on $\T_1$ under the random cluster model. Now, suppose that $\T_2$ is a binary phylogenetic $X$--tree that is different from  $\T_1$. By Lemma~\ref{decomplem}, $\T_2$ and $\T_1$ both share the same pair of  pendant subtrees $\T_A$ and $\T_B$ for which the roots are adjacent in $\T_2$ but not in $\T_1$, as illustrated in Fig. 3. Now consider the evolution of one of the characters $\chi_i$ under the random cluster model on $\T_1$. Let $\alpha, \beta$ denote, respectively, the blocks of the character at vertices $u, v$ of $\T_1$, and consider the conjunctive event $E = \bigcap_{i=1}^5E_i$ in which:
\begin{itemize}
\item[($E_1$)] $\alpha \neq \beta$;
\item[($E_2$)] at least one leaf of $\T_A$ is present in $\alpha$;
\item[($E_3$)] at least one leaf of the other subtree of $\T_1$, which is incident with $u$ and does not contain $v$, is present in  $\alpha$;
\item[($E_4$)]  at least one leaf of $\T_B$ is present in $\beta$; and
\item[($E_5$)] at least one leaf of the other subtree of $\T_1$, which is incident with $v$ and does not contain $u$, is present in $\beta$.
\end{itemize}
Under the random cluster model, these five events are independent (by the assumption that the cuts on edges are made independently) and so:
$$\PP(E) = \PP(\bigcap_{i=1}^5E_i) = \prod_{i=1}^5\PP(E_i) \geq \gamma,$$
since $\PP(E_1) \geq a$, and, by Lemma 2.1 of \cite{cluster},
$\PP(E_i) \geq (1-2b)/(1-b)$ for $i \in \{2,\ldots, 5\}$.
Now, $\psi_H(\F, \T_2) = {\rm `false}$' whenever event $E$ occurs.  If we evolve $k$ independent characters under the assumptions stated in the Theorem, then the probability that $E$ occurs at least once is at least $1-(1-\gamma)^k$, and this probability is at least $1-\epsilon$ when $k \geq \gamma^{-1}\log(\epsilon^{-1})$, by virtue of the inequality $1 - (1-x)^y \geq 1- e^{-xy}$. This completes the proof.
 \hfill $\Box$

\section{Concluding comments}
\label{conc}

The reader may be curious as to where our proof for the $\log(n)$ lower bound on sequence length for testing under the finite-state model breaks down for the random cluster model.  The crucial distinction is that the random cluster model fails to satisfy condition (\ref{tau}) required in the proof for Theorem~\ref{finitestate}.  That is, in the random cluster model,  some characters have a positive probability on some trees but have zero probability on other trees.  Indeed it has been shown that, for any binary phylogenetic tree $\T$ with $n$ leaves, there is a set of just {\em four} characters such that $\T$ is the only tree for which these characters have strictly positive probability \cite{hube}.  Thus, in contrast to finite-state models, under the random cluster model each tree can be reconstructed from $O(1)$ characters (using, say maximum likelihood estimation), provided these characters are carefully selected; if the characters evolve under the random cluster model then, as mentioned earlier, the number required number of characters for accurate tree reconstruction grows at the rate of (at least) $\log(n)$ \cite{cluster}.

Note also that Theorem~\ref{randyclust} can be extended to  a setting in which the substitution probabilities vary from character to character, provided they all lie in some fixed interval $[a,b]$ where $0<a \leq b<\frac{1}{2}$.  If we generate $k$ characters independently (but not  necessarily identically) in this more general setting, testing the true tree using $\psi_H$ will return `true' with probability  $1$, while testing any other tree will return `false' with probability at least $1-\epsilon$ provided $k$ satisfies the lower bound described in Theorem~\ref{randyclust}.

Although testing for the  finite-state Markov model can require the same $\Omega(\log(n))$ growth in sequence length required for reconstructing, there is a related task where $O(1)$ sequence length suffices. This is for {\em teasing a tree}, where one is given sequences of length $k$ and a set of two trees -- one of which is the tree that generated the data, and one is asked to identify which of the two trees generated the data. For the symmetric $2$--state Markov process and under suitable restrictions on the substitution probabilities, it was shown in \cite{teasing} that sequences of length $O(1)$ (i.e. independent of $n$) suffice to correctly solve (with high probability) the teasing problem on binary phylogenetic trees with $n$ leaves.

Finally, we have considered reconstruction only in the part of the parameter range (on the substitution probabilities on the edges of the tree) where reconstruction requires logarithmic length sequences. Outside of this region, it is known that polynomial-length sequences can be required, both for the finite state Markov model \cite{logs3} and for the random cluster model \cite{cluster}. It may  be of interest in future work to determine the sequence length required for testing in these portions of parameter space.


\newpage

\begin{thebibliography}{99}

\bibitem{ber}
Berger, A., Wald, A., 1949.
On distinct hypotheses.
Ann. Math. Stat. 20(1), 104--109.

\bibitem{chu}
Churchill, G., von Haeseler, A., Navidi, W., 1992.
Sample size for a phylogenetic inference.
Mol. Biol. Evol. 9(4), 753--769.

\bibitem{logs3}
Daskalakis, C., Mossel E., Roch, S., 2006.
Optimal Phylogenetic Reconstruction,
in Proceedings of the thirty-eighth annual ACM symposium on Theory of computing (STOC 2006). 159--168.

\bibitem{logs}
Erd{\"o}s, P.L., Steel, M.A., Székely, L.A., Warnow, T., 1999.
A few logs suffice to build (almost) all trees (Part 1).
Random Struct. Algorithms 14(2), 153--184.

\bibitem{fels}
Felsenstein, J., 2003.
Inferring phylogenies.
Sinauer Press.


\bibitem{lec}
Lecointre, G., Philippe, H., Van Le, H.L., Le Guyader, H., 1994.
How many nucleotides are required to resolve a phylogenetic problem? The use of a new statistical method applicable to available sequences.
Mol. Phyl. Evol. 3(4), 292--309.

\bibitem{cluster}
Mossel, E., Steel, M., 2004.
A phase transition for a random cluster model on phylogenetic trees.
Math. Biosci. 187, 189--203.

\bibitem{hube}
Huber, K., Moulton, V., Steel, M., 2005.
Four characters suffice to convexly define a phylogenetic tree.
SIAM J. Discrete Math. 18(4), 835--843.


\bibitem{mos1}
Mossel, E., Steel, M., 2005.
How much can evolved characters tell us about the tree that generated them? In: Olivier Gascuel (ed.),
Mathematics of Evolution and Phylogeny, Oxford University Press, 384--412.

\bibitem{sai}
Saitou, N., Nei, M., 1986.
The number of nucleotides required to determine the branching order of three species, with special reference to the human-chimpanzee-gorilla divergence. J. Mol. Evol. 24, 189--204.

\bibitem{sem}
Semple, C., Steel, M., 2003.
Phylogenetics.
Oxford University Press.

\bibitem{inv1}
Steel, M.A., Sz{\'e}kely, L.A., 1999.
Inverting random functions.
Ann. Comb. 3, 103--113.

\bibitem{inv2}
Steel, M.A., Sz{\'e}kely, L.A., 2002.
Inverting random functions II: explicit bounds for the discrete maximum likelihood estimation, with applications.
SIAM J. Discrete Math. 15(4), 562--575.

\bibitem{inv3}
Steel, M.A., Sz{\'e}kely, L.A., 2008.
Inverting random functions III: Discrete maximum likelihood revisited.
Ann. Comb, in revision, (preprint  arXiv:math/0608273).

\bibitem{teasing}
Steel, M.A., Sz{\'e}kely, L.A., 2007.
Teasing apart two trees.
Comb. Probab. Comput. 16, 903--922.

\bibitem{wort}
Wortley, A.H., Rudall, P.J., Harris, D.J., Scotland, R.W., 2005.
How much data are needed to resolve a difficult phylogeny? Case study in Lamiales.
Syst. Biol. 54(5), 696--709.

\end{thebibliography}
\end{document}